\documentclass{llncs}

\usepackage{graphicx,fancyvrb}

\usepackage{latexsym}
\usepackage{amssymb}
\usepackage{amsmath}
\usepackage{amsfonts}
\usepackage{proofs}
\usepackage{url}

\usepackage{color}
\usepackage{algorithmic}
\usepackage{algorithm}

\newcommand{\trust}[0]{\triangleleft}

\newcommand{\Vars}{\mathit{Vars}}
\newcommand{\Lit}{\mathit{Lit}}

\newcommand{\Agt}{\mathit{Ag}}
\newcommand{\Temp}{\mathit{T}}

\newcommand{\Hist}{\mathcal{EL}}
\newcommand{\Theory}{\mathcal{E}}
\newcommand{\Interp}{\mathfrak{I}}
\newcommand{\hfunc}{\mathcal{F}}
\newcommand{\po}[0]{\mathcal{PO}}

\newcommand{\trustTrans}{\textsc{Trans}\trust}
\newcommand{\RtrustTrans}{\textsc{Trans}\prec}

\newcommand{\ddoc}[3]{#1: (#2: #3)}
\newcommand{\idoc}[2]{#1\ [#2]}

\newcommand{\wffHLzero}{\rho}
\newcommand{\wffHLone}{\varphi}

\newcommand{\disc}[0]{\mathcal{D}}
\newcommand{\trule}[0]{\mathcal{L}}
\newcommand{\secrule}[0]{(\rightarrow)}
\newcommand{\crule}[0]{\mathcal{C}}

\newcommand{\TheForensicator}{\mathit{TF}}
\newcommand{\FireEye}{\mathit{FE}}
\newcommand{\CrowdStrike}{\mathit{CS}}

\begin{document}

\title{A Formal Approach to Analyzing Cyber-Forensics Evidence}

\author{Erisa Karafili\inst{1} \and Matteo Cristani\inst{2} \and Luca Vigan\`o\inst{3}}
\institute{Department of Computing, Imperial College London, UK\\
\email{e.karafili@imperial.ac.uk}
\and 
Dipartimento di Informatica, Universit\`a di Verona, Italy \\
\email{matteo.cristani@univr.it}
\and
Department of Informatics, King's College London, UK\\
\email{luca.vigano@kcl.ac.uk}
}

\maketitle

\begin{abstract}
The frequency and harmfulness of cyber-attacks are increasing every day, and with them also the amount of data that the cyber-forensics 
analysts need to collect and analyze.
In this paper, we propose a formal analysis process that allows an analyst to filter the enormous amount of evidence collected  and either identify crucial information about the attack (e.g., when it occurred, its culprit, its target) or, at the very least, perform a pre-analysis to reduce the complexity of the problem in order to then  draw conclusions more swiftly and efficiently. 

We introduce the Evidence Logic $\Hist$ for representing simple and derived pieces of evidence from different sources. 
We propose a procedure, based on monotonic reasoning, that rewrites the pieces of evidence with the use of tableau rules,
based on relations of trust between sources and the reasoning behind the derived evidence, and yields a consistent set of pieces of evidence. 
As proof of concept, we apply our analysis process to a concrete cyber-forensics case study.

\end{abstract}
\section{Introduction}

The frequency and harmfulness of cyber-attacks are increasing every day, and with them also the amount of data that cyber-forensics 
analysts need to collect and analyze. In fact, forensics investigations often  produce an enormous amount of evidence. The pieces of evidence are produced/collected by various sources, which can be humans (e.g., another analyst) or forensic tools such as intrusion detection system (IDS), traceback systems, malware analysis tools, and so on. 

When a forensics analyst evaluates a cyber-attack, she first collects all the relevant evidence containing information about the attack and then checks the sources of the evidence in order to evaluate their reliability and to resolve possible inconsistencies arising from them. 
Based on the collected information, which might be different depending on the information sources and the trust relation between the analyst and the sources, the analyst might reconstruct different, possibly faulty, courses of events. 
State of the art approaches don't really manage to cope well with such situations.

To reason about the collected evidence, we need to formalize the fact that the analyst trusts more some sources than others for particular pieces of evidence, e.g., a source $S_1$ is more trusted than another source $S_2$ for attack similarity as tool $S_1$ specializes in malware analysis whereas tool $S_2$ specializes in deleted data. We also need to distinguish between the evidence and its interpretation that an analyst may consider in order to perform a correct analysis and attribution of the cyber-attack.

Our main contribution in this paper is the introduction of the \emph{Evidence Logic $\Hist$}, which allows an analysts to represent the different pieces of evidence, together with their  sources and relations of trust, and reason about them by eliminating the conflicting pieces of evidence during the analysis process. 

As a concrete motivating example, consider the data breach of the Democratic National Committee (DNC) network, during the last US presidential campaign, when Wikileaks and other websites published several private emails in October and November 2016. 
DNC used the services of a cyber-security company, CrowdStrike,
to mitigate the attacks and to conduct a forensics investigation.
CrowdStrike stated that the main attack occurred between March and April 2016, and identified it as a \emph{spear phishing} campaign that used Bitly accounts to shorten malicious URLs. 
The phishing campaign was successful as different IDs and passwords were collected. 

However, TheForensicator, an anonymous analyst, stated that the attack actually occurred on the 5th of July 2016, not in March/April, as the metadata released by an alleged attacker were created on the 5th of July 2016, and the data-leak occurred physically as the data were transferred at the speed of around 23 MB/s, and this speed is possible only during a physical access. 
Another cyber-security company, FireEye, stated that it is possible to have a non-physical data-transfer speed of 23 MB/s. 
What should an analyst conclude from these discording statements and pieces of evidence? How can a decision be made?

$\Hist$ is able to deal with this type of discordances, and based on relations of trust on the sources and reasonings, to arrive at a certain conclusion. 
$\Hist$ is composed of three separate layers: the first layer $\Hist_E$ deals with pieces of evidence, the second layer $\Hist_I$ focuses on the evidence interpretations, and the third layer $\Hist_R$ focuses on the reasoning involved in the evidence. Reasoning with $\Hist$ amounts to applying a rewriting system that spans formulas in all three levels to reach a conclusion, ruling out discordances and inconsistencies. Applying the reasoning process of $\Hist$ to the different pieces of evidence from the various sources, the analyst can decide the type of the attack and when it occurred. For instance, regarding the speed of transferability, if the analyst trusts FireEye more than TheForensicator, then she does not take into consideration the evidence that the data transfer was physical. Hence, she concludes that the attack occurred during March/April 2016 and not in July 2016.

We proceed as follows. In Section~\ref{sec:frame} and Section~\ref{sec:semantics}, we give the syntax and semantics of the Evidence Logic $\Hist$, respectively. In Section~\ref{sec:rewriting}, we introduce the rules of the rewriting system of $\Hist$ and we give a concrete procedure that uses the rewriting rules to prove the satisfiability of a given $\Hist$-theory (which is a finite and non-empty set of formulas of the three layers of $\Hist$). We prove the rewriting system to be sound and the procedure to be correct 
(the proofs of the theorems are given in the Appendix).
Section~\ref{sec:related} concludes the paper by discussing related work and future work.

\section{The Syntax of the Evidence Logic $\Hist$}
\label{sec:frame}

The \emph{Evidence Logic} $\Hist$ that we propose enables a cyber-forensics analyst to represent the various plausible pieces of evidence that she collected from different sources and reason about them. To that end, the analyst should distinguish between the evidence and its interpretation. 
In a nutshell:
\begin{itemize}
\item \emph{evidence} represents information related to the attack, where a given (piece of) evidence usually represents an event, its occurrence and the source of the information of the occurrence of the event (which can be another analyst, a cyber-forensics tool, etc.);
\item \emph{evidence interpretation}
represents what the analyst thinks\footnote{We deliberately use the verb ``thinking'' to avoid suggesting any epistemic or doxastic flavor, as in $\Hist$ we do not 
consider the modalities of knowledge or belief.} about the occurrence of an event $e$ and about the occurrences of the  
events causing $e$.
\end{itemize}
$\Hist$ contains two types of well-formed formulas: \emph{labeled formulas}, to formalize the pieces of evidence and interpretation, and \emph{relational formulas}, to formalize relations of \emph{trust} between sources of evidence and their reasonings. 
$\Hist$ also contains a rewriting system (composed of a set of tableau rules) to build the analyst's interpretations from forensics evidence.
For the sake of readability, we omit to model explicitly the analyst who is reconstructing and attributing the cyber-attack, but we simply silently assume her existence.

$\Hist$ is composed of three separate layers: the first layer $\Hist_E$ shows how the well-formed formulas for pieces of evidence are built, the second layer $\Hist_I$ focuses on the evidence interpretations,
and the third layer $\Hist_R$ focuses on the reasoning involved in the evidence. In the following, we discuss each of the layers in detail.

\subsection{$\Hist_E$: Evidence}
\label{subsec:hist-doc}

\begin{definition}
Given $t,t_1,\dots t_n \in \Temp$, $a,a_1, \dots a_n \in \Agt$, $r_1, r_2 \in \mathcal{R}$, $p \in \Vars_S$ and $\phi, \phi_1,\allowbreak \ldots,\allowbreak \phi_n \in \Lit$,
the set $\wffHLzero$ of
formulas of $\Hist_E$ is 
\begin{displaymath}
\begin{array}{l}
\wffHLzero ::=  
\ddoc{a}{t}{\phi} 
\; \big| \;
\idoc{\ddoc{a}{t}{\phi}}{\ddoc{a_1}{t_1}{\phi_1} \mid \ldots \mid \ddoc{a_n}{t_n}{\phi_n}}_r
\; \big| \;
a_1 \trust_p a_2 
\; \big| \;
r_1 \prec r_2 
\end{array}
\end{displaymath}
\end{definition}

We introduce all these notions, and the four kinds of formulas, step by step. 
A piece of evidence asserts what a source thinks about the temporal
occurrence of an event, i.e., whether an event occurred or not 
in a particular instance of time. 
To formalize this, we use two finite\footnote{In principle, there is
nothing in our logic that prevents us from considering countable,
possibly infinite, sets of labels, but here we consider finite sets for
simplicity.} and disjoint sets of labels,
\begin{itemize}
\item \emph{source labels} $\Agt=\{a_1, a_2, \ldots , a_n\}$ for
forensic sources, which we call \emph{agents}, regardless of whether they are human or not, and
\item \emph{temporal labels} $\Temp=\{t_1, t_2, \ldots ,t_m\}$ for
instants of time,
\end{itemize}
along with 
\begin{itemize}
\item a set of \emph{propositional variables} $\Vars = \{p_1, p_2, \ldots, p_n\}$ that represent the occurrences of \emph{forensics events} (so that $p$ represents the occurrence of an event and $\neg p$ represents that $p$ does not occur), 
\item a set of \emph{reasoning rules} (or simply \emph{reasonings}) $\mathcal{R}= \{r_1, r_2, \ldots, r_l\}$ that represent the reasoning used by the agents to conclude further evidence.
\end{itemize}
The set of \emph{literals} $\Lit = \{p_1, \neg p_1, 
\ldots, p_n, \neg p_n\}$ consists of each propositional variable and its negation. We write $\phi$, possibly subscripted, to denote a literal.

Instants of time are labels associated to elements of a single given stream. Thus, the labels that represent the instants of time cannot be processed for consistency, and no
assertions regarding relations between them is allowed.

\begin{example}
Consider again the motivating example that we discussed in the Introduction. The set of agents is composed of the analyst (whose existence we silently assume) and the sources
CrowdStrike ($\CrowdStrike)$, TheForensicator ($\TheForensicator$) and
FireEye ($\FireEye$); thus, $\Agt = \{\CrowdStrike, \TheForensicator, \FireEye\}$.
The sources make statements about events occurring in two instants of time: ``March/April 2016'' and ``5th of July 2016'' represented respectively by $t_1$ and $t_2$.
\hfill $\Box$
\end{example}

We formalize two different types of evidence: simple and derived one. 
The \emph{simple evidence} is a labeled formula of the form
\begin{displaymath}
\ddoc{a}{t}{\phi}\,,
\end{displaymath}
expressing that the agent represented by the source label $a$ thinks
that the literal $\phi$ is true at the instant of time represented by the temporal label $t$. For short, we will say that $a$ thinks that $\phi$ is true at $t$.

\begin{example}
The simple evidence $\FireEye: (t_2: \mathit{SpeedTr(23MB/s))}$ expresses that $\FireEye$ states that the non-physical transferability speed, $\mathit{SpeedTr}$, can be $\mathit{23MB/s}$ at $t_2$. 
\hfill $\Box$
\end{example}

The \emph{derived evidence} is a labeled formula of the form
\begin{displaymath}
\idoc{\ddoc{a}{t}{\phi}}{\ddoc{a_1}{t_1}{\phi_1} \mid \ddoc{a_2}{t_2}{\phi_2} \mid 
\ldots \mid \ddoc{a_n}{t_n}{\phi_n}}_r\,,
\end{displaymath}
expressing that agent $a$ thinks that $\phi$ is true at
instant of time $t$ \emph{because} of reasoning 
$r$, where $a_1$ thinks that $\phi_1$ is true at
$t_1$, $\ldots$ and $a_n$ thinks that $\phi_n$ is true at $t_n$. In other words, based on $r$, $a$ thinks 
that $\phi$ is \emph{caused}\footnote{We use the term ``cause'' to describe the events that an agent thinks were the preconditions for 
a certain derived evidence. In this work, we will not focus on the causality relationships between events.} 
by $\phi_1, \cdots, \phi_n$ 
(with their respective time instants and agents).
The reasoning $r$ of the derived evidence $a:(t:\phi)$ is composed of simple and/or derived pieces of evidence.
We include a constraint in our syntax that does not permit cycles
between derived pieces of evidence, so that if $\idoc{\ddoc{a_i}{t_i}{\phi_i}}{\cdots \mid \ddoc{a_j}{t_j}{\phi_j} \mid \ldots}\,_r$, then we do not accept in our language the formula $\idoc{\ddoc{a_j}{t_j}{\phi_j}}{\cdots \mid \ddoc{a_i}{t_i}{\phi_i} \mid \ldots}\,_{r'}$. 

A reasoning $r$ can be used by different agents to arrive at the same conclusion (derived evidence), using the same pieces of evidence. An agent can use different reasonings $r_i, \cdots, r_j$ to conclude the same derived evidence, where the pieces of evidence used by the reasonings are different from one reasoning to the other.

\begin{example}
$\CrowdStrike$ says that the $\mathit{Attack}$ occurred at time $t_1$,
based on reasoning $r_1$ and $\CrowdStrike$'s evidence about a spear phishing campaign $\mathit{SpPhish}$ and its success $\mathit{SucPhish}$ at $t_1$. The latter is based on $r_2$ and $\CrowdStrike$'s pieces of evidence that in $t_1$: the malicious link was clicked $\mathit{LinkCl}$, the malicious form was filled $\mathit{FFill}$, 
 and the data were stolen $\mathit{DStolen}$. We thus have:
\begin{displaymath}\small
\begin{array}{l}
\idoc{\ddoc{\CrowdStrike}{t_1}{\mathit{Attack}}}{\ddoc{\CrowdStrike}{t_1}{\mathit{SpPhish}} \mid \ddoc{\CrowdStrike}{t_1}{\mathit{SucPhish}}}_{r_1} \\
\idoc{\ddoc{\CrowdStrike}{t_1}{\mathit{SucPhish}}}{\ddoc{\CrowdStrike}{t_1}{\mathit{\mathit{LinkCl}}} \mid \ddoc{\CrowdStrike}{t_1}{\mathit{FFill}} \mid \ddoc{\CrowdStrike}{t_1}{\mathit{DStolen}}}_{r_2} 
\end{array}
\end{displaymath}
Instead, $\TheForensicator$ says that based on $r_3$ the $\mathit{Attack}$ occurred at $t_2$ because the metadata $\mathit{MetaC}$ were created at $t_2$ and the access was physical $\mathit{PhysA}$. The latter is true because $\TheForensicator$ states that it is not true that $\mathit{SpeedTr}$ is $\mathit{23MB/s}$:
\begin{displaymath}\small
\begin{array}{l}
\idoc
{\ddoc{\TheForensicator}{t_2}{\mathit{Attack}}}
{\ddoc{\TheForensicator}{t_2}{\mathit{MetaC}} 
\mid 
\ddoc{\TheForensicator}{t_2}{\mathit{PhysA}}}_{r_3} \\
\idoc{\ddoc{\TheForensicator}{t_2}{\mathit{PhysA}}}{\ddoc{\TheForensicator}{t_2}{\mathit{\neg \mathit{SpeedTr(23MB/s)}}}}_{r_4} 
\end{array}
\end{displaymath}
\hfill $\Box$
\end{example}

To allow an analyst to distinguish the events that can be expressed by a simple or derived evidence, the set of propositional variables $\Vars$
is composed by two disjoint subsets $\Vars_S$ and $\Vars_D$ that respectively represent the events that can be part of simple and 
derived evidence, i.e., $\Vars = \Vars_S \cup \Vars_D$ with $\Vars_S\cap  \Vars_D = \emptyset$. By extension, we write $\phi \in \Lit_S$ if $\phi$ is $p$ or $\neg p$ with $p \in \Vars_S$, and $\phi \in \Lit_D$ if $\phi$ is $p$ or $\neg p$ with $p \in \Vars_D$.

Hence, if $\phi \in \Lit_S$, then $a:(t:\phi)$ is a simple evidence, whereas if $\phi \in \Lit_D$ and $\phi_i \in \Lit$ for $i \in \{1,\ldots,n\}$, then $\idoc{\ddoc{a}{t}{\phi}}{\ddoc{a_1}{t_1}{\phi_1} \mid \ldots \mid \ddoc{a_n}{t_n}{\phi_n}}_r$ is a derived evidence.
For simplicity, 
we will assume that a variable that represents an event given by a simple evidence is part of $\Vars_S$ and that a variable that represents an event given by a derived evidence is part of $\Vars_D$.

\begin{example}
The variables of the events of our example are divided in the two following disjoint subsets:
$\Vars_S = \{\mathit{SpPhish}, \mathit{LinkCl}, \mathit{FFill}, \mathit{DStolen}, \mathit{MetaC}, \allowbreak \mathit{SpeedTr(23MB/s)}\}$  and $
\Vars_D = \{\mathit{Attack}, \mathit{SucPhish}, \mathit{PhysA}\}$.
\hfill $\Box$
\end{example}

The temporal labels can have \emph{temporal constraints} such as $t_1 \leq t$ or $t_n < t$. As we consider time to be \emph{linear} and every instant of time is mapped to only one element of the natural numbers,
our syntax doesn't need to include a precedence relation, as it represents the classical precedence relation between natural numbers.

In addition to ordering events with respect to time, the analyst can
consider the \emph{trust(worthiness)} relations that she has with the sources with respect to their assertions in the simple evidence, i.e., she might think that one source is more reliable than another one with respect to a particular event (and its negation). For instance, $a_i$ might be more trustworthy than $a_j$ with respect to an event $p$ (and thus also $\neg p$), where $p\in \Vars_S$. In general, if there exists a trust relation between two agents $a_i, a_j \in \Agt$ for an event $p \in \Vars_S$, then we have that either
$a_i$ is more trustworthy than $a_j$ with respect to $p$, or
$a_j$ is more trustworthy than $a_i$ with respect to $p$. 
We formalize this by introducing the \emph{trust relation $\trust : \Agt
\times \Agt \times \Vars_S$}. 
Then, the \emph{relational formula} $a_i \trust_p a_j$ expresses that $a_j$ is more trustworthy than $a_i$ with respect to $p$.

\begin{example}
We write $\TheForensicator \trust_{\mathit{SpeedTr(23MB/s)}} \FireEye$ to formalize that the analyst trusts $\FireEye$ more than $\TheForensicator$ w.r.t.~the simple evidence $\mathit{SpeedTr(23MB/s)}$.
\hfill $\Box$
\end{example}

The analyst can also consider the trust(worthiness) relations about the reasonings she used. In particular, given two conflicting derived pieces of evidence that use two different reasonings, the analyst can consider one reasoning to be more trustworthy than the other one. We formalize this by introducing the \emph{trust relation} $\prec: \mathcal{R} \times \mathcal{R}$. Then, the \emph{relational formula $r_i \prec r_j$} expresses that reasoning $r_j$ is more trustworthy than reasoning $r_i$.

\subsection{$\Hist_I$: Evidence Interpretation}
\label{subsec:hist-inter}

An \emph{evidence interpretation} (or simply \emph{interpretation})
is what the cyber-forensics analyst thinks that is plausibly true. To formalize this, the second level $\Hist_I$ of $\Hist$ employs a simplified variant of \emph{Linear Temporal Logic (LTL)}.
$\Hist_I$  inherits from $\Hist_E$ the temporal labels $\Temp$, the reasonings $\mathcal{R}$ and the propositional variables $\Vars$ (and thus also the literals $\Lit$).

\begin{definition}
Given $t, t_1, \ldots t_n \in \Temp$, 
$\phi, \phi_1, \ldots, \phi_n \in \Lit$, $r \in \mathcal{R}$ and 
$\phi' \in \Lit_D$, the set $\wffHLone$ of 
formulas of $\Hist_I$, called \emph{interpretations}, is 
\begin{displaymath}
\wffHLone ::= \ 
t:\phi \mid t_1: \phi_1 \wedge t_2:\phi_2 \wedge \ldots \wedge t_n:\phi_n \rightarrow_r t:\phi'
\end{displaymath}
\end{definition}

$t:\phi$ means that the analyst thinks that $\phi$ is true at $t$, whereas $t_1: \phi_1 \, \wedge \ldots \wedge \, t_n:\phi_n \rightarrow_r t:\phi'$ means that the analyst thinks that $\phi'$ is true at the instant of time $t$, based on reasoning $r$, if $\phi_i$ is true at $t_i$ for all $i \in \{1, \ldots, n\}$.
An interpretation expresses a positive event (the occurrence of an event, e.g., $t:p$) or a negative event (the non occurrence of an event, e.g., $t:\neg p$). The interpretations of the temporalized logic $\Hist_I$ that express positive events represent the \emph{plausible pieces of evidence} and help the analyst to perform a correct analysis.

\subsection{$\Hist_R$: Evidence Reasoning}\label{subsec:reasoning}

The third layer $\Hist_R$ of $\Hist$ is the \emph{reasoning layer} and deals with the reasoning behind the derived evidence. Also $\Hist_R$ uses LTL and inherits from $\Hist_E$ the temporal labels $\Temp$, the reasonings $\mathcal{R}$ and the  propositional variables $\Vars$.
\begin{definition}
Given $t \in \Temp$, $\phi \in \Lit_D$ and $r, r_k, \ldots, r_l \in \mathcal{R}$,
the set $\psi$ of 
formulas of $\Hist_R$ is 
\begin{displaymath}
\psi ::= \ 
 (t:\phi)_r \mid (t:\phi)_{r, r_k, \ldots, r_l}.
\end{displaymath}
\end{definition}

The \emph{reasoning} involves only derived pieces of evidence, 
which we can divide in two types.
The \emph{first type of derived evidence}, $(t:\phi)_r$, is composed of only simple pieces of evidence; in this case, the only reasoning is the one made by the agent that states the derived evidence $\idoc{\ddoc{a}{t}{\phi}}{\ddoc{a_1}{t_1}{\phi_1} \mid \ldots \mid \ddoc{a_j}{t_j}{\phi_j}}_r$, where $\phi_i \in \Lit_S$ for $i \in \{1,\ldots,j\}$.
The \emph{second type of derived evidence}, $(t:\phi)_{r, r_k, \cdots, r_l}$, is composed of simple and derived pieces of evidence; in this case, the reasoning involves the one of the agent stating the derived evidence, $\idoc{\ddoc{a}{t}{\phi}}{\ddoc{a_1}{t_1}{\phi_1} \mid \ldots \mid \ddoc{a_j}{t_j}{\phi_j}}_r$, as well as all the reasonings involved in the derived pieces of evidence $\phi_i \in \Lit$ for $i \in \{1,\ldots,j\}$ that are part of reasoning $r$.
The first type is clearly a special case of the second one, but we keep both for the sake of understandability.

\section{The Semantics of the Evidence Logic $\Hist$}
\label{sec:semantics}

\begin{definition}
The \emph{plausible pieces of evidence} are a finite stream of temporal instants in which at every instant of time we may associate a finite number of occurrences or not occurrences of an event. 
\end{definition}

The agents are associated to a given finite set of values, and the trust relationship between agents is interpreted as a partial order on the agents. The same applies to the reasonings: they are associated to a finite set of values and the trust relationship between them is interpreted as a partial order on the reasonings.

\begin{definition}\label{semantic}
A model of the evidence language $\Hist$ is a tuple
\begin{displaymath}
\mathfrak{M}=\{\Agt^\Interp, \hfunc^\Interp, \po^\Interp, \mathcal{TR}^\Interp, \Vars^\Interp, \mathcal{R}^\Interp, \Interp \}
\end{displaymath}
where: 
\begin{itemize}
\item $\Interp$ is the interpretation function, where we interpret time as natural numbers, i.e., $t^\Interp \in \mathbb{N}$ for every $t \in \Temp$;
\item $\Agt^\Interp = \{{a_1}^\Interp, \ldots {a_n}^\Interp\} = \{a_1, \ldots a_n\} = \Agt$ is a set of agents;
\item $\hfunc^\Interp$ is a function that maps pairs of instants of time and formulas to $\mathit{True}$ or $\mathit{False}$ (this mapping is used in the second layer of $\Hist$, where we have $t:\phi$);
\item $\po^\Interp = \{{\trust_{p_i}}^\Interp\}$ is a set of trust relationships between agents, where for every $p\in \Vars_S$, if ${\trust_p}^\Interp \in \po^\Interp$, then ${\trust_p}^\Interp = \{({a_i}^\Interp, {a_j}^\Interp)\mid a_i \trust_p a_j\}^{*}$, where $*$ is the transitive closure of $\trust$;
\item $\mathcal{TR}^\Interp = \{\prec^\Interp\}$ is a set of trust relationship between reasonings, where for every $r\in \mathcal{R}$, if ${\prec}^\Interp \in \mathcal{TR}^\Interp$, then ${\prec}^\Interp = \{({r_i}^\Interp, {r_j}^\Interp)\mid r_i \prec r_j\}^{*}$,
where $*$ is the transitive closure of $\prec$;
\item $\Vars^\Interp = \Vars = \{p_1, \cdots, p_n\}$ is a set of events;
\item $\mathcal{R}^\Interp = \mathcal{R} = \{r_1, r_2, \cdots, r_n\}$ is a set of reasoning rules.
\end{itemize}
\end{definition}

Slightly abusing notation, we use $\Agt^\Interp$ to denote also a set of functions, each function ${a_i}^\Interp:\mathbb{N}\times \Lit \rightarrow \{\mathit{True}, \mathit{False}\}$ associating to an instant of time $t$ a set of formulas that are true at $t$, where
${a_i}^\Interp(t, p) = \mathit{True}$ when $\ddoc{a_i}{t}{p}$ is asserted,
${a_i}^\Interp(t, p) = \mathit{False}$ when $\ddoc{a_i}{t}{\neg p}$ is asserted,
${a_i}^\Interp(t, \neg p) = \mathit{True}$ when $\ddoc{a_i}{t}{\neg p}$ is asserted,
${a_i}^\Interp(t, \neg p) = \mathit{False}$ when $\ddoc{a_i}{t}{p}$ is asserted.
The same applies to $\mathcal{R}^\Interp$, each function ${r_i}^\Interp:\mathbb{N}\times \Lit \rightarrow \{\mathit{True}, \mathit{False}\}$ such that
$(t, p)_{{r_i}^\Interp} = \mathit{True}$ when $(t:p)_{r_i}$ is asserted,
$(t, p)_{{r_i}^\Interp} = \mathit{False}$ when $(t:\neg p)_{r_i}$ is asserted,
$(t, \neg p)_{{r_i}^\Interp} = \mathit{True}$ when $(t: \neg p)_{r_i}$ is asserted,
$(t, \neg p)_{{r_i}^\Interp} = \mathit{False}$ when $(t:p)_{r_i}$ is asserted.
Thus, ${a_i}^\Interp$ and 
$\hfunc^\Interp$
 both associate to every $t$ a set of formulas that are true at $t$; 
 the difference is that we use the ${a_i}^\Interp$ in the evidence layer $\Hist_E$ and 
 $\hfunc^\Interp$
  in the  interpretation layer $\Hist_I$.

In order to avoid having clear contradictions in the models, we constrain the functions $\Agt^\Interp$ and $\mathcal{R}^\Interp$ as follows:
\begin{itemize}
\item[] $(\mathit{COND}_1)$: If ${a}^\Interp(t, p) = \mathit{True}$, then $a^\Interp (t', p)= \mathit{False}$ for all $t'\neq t$.
\item[] $(\mathit{COND}_2)$: If $(t, p)_{r^\Interp} = \mathit{True}$, then $(t', p)_{r^\Interp}= \mathit{False}$ for all $t'\neq t$.
\item[] $(\mathit{COND}_3)$: Every ${\trust_p}^\Interp$ is an irreflexive and antisymmetric relation.
\item[] $(\mathit{COND}_4)$: Every ${\prec}^\Interp$ is an irreflexive and antisymmetric relation.
\end{itemize}

A \emph{$\Hist$-theory} is built by using $\Hist$ to express a finite and non-empty set of formulas 
of the three layers,
 including the trust relationships.

\section{The Rewriting System of the Evidence Logic $\Hist$}
\label{sec:rewriting}

\begin{table}[!t]
\caption{Rules of the rewriting system of $\Hist$\label{table:rules}}
\begin{center}
\scalebox{1}{
\fbox{
\begin{tabular}{c}
$\infer[\trule_1]{
\Theory \cup \{t:\phi\}
}{
\ddoc{a}{t}{\phi}
}
$
\qquad \quad \quad
$\infer[\trule'_1]{
\Theory\cup \{t:\phi\}
}{
(t:\phi)_{r, \cdots, r_n}
}$\\

\\
$
\infer[\trule_2]{
\Theory \cup \{\ddoc{a_i}{t_i}{\phi_i}\}_{\forall i\in\{1, \cdots, n\}\   \phi_i \in \Lit_S
}\cup 
\{t_1:\phi_1\wedge \dots \wedge t_n:\phi_n\rightarrow_r t:\phi\}
}{
\idoc{\ddoc{a}{t}{\phi}}{\ddoc{a_1}{t_1}{\phi_1} \mid \dots
\mid \ddoc{a_n}{t_n}{\phi_n}}_r
}$\\

\\ 
$
\infer[(\rightarrow)]{
\Theory \cup \{(t:\phi)_r\}
}{
t_1:\phi_1 \wedge \cdots \wedge t_n : \phi_n \rightarrow_r t:\phi &
t_1:\phi_1 & \cdots & t_n:\phi_n
}$\\

\\
$
\infer[(\rightarrow')]{
\Theory \cup \{(t:\phi)_{r, r_1/ \emptyset, \cdots, r_n/ \emptyset\}}
}{
t_1:\phi_1 \wedge \cdots \wedge t_n : \phi_n \rightarrow_r t:\phi &
(t_1:\phi_1)_{r_1/ \emptyset} & 
\cdots & (t_n:\phi_n)_{r_n/ \emptyset}
}$\\

\\
$
\infer[\trustTrans]{
\Theory \cup \{a_1 \trust_p a_3\}
}{
a_1 \trust_p a_2 & a_2\trust_p a_3
}
\qquad \qquad 
\infer[\RtrustTrans]{
\Theory \cup \{r_1 \prec r_3\}
}{
r_1 \prec r_2 & r_2\prec r_3
}$\\

\\
$\infer[\disc_1 {[*]}]{
\Theory \cup \{\ddoc{a_1}{t_2}{\neg \phi}, \ddoc{a_2}{t_1}{\neg \phi}\}
}{
\ddoc{a_1}{t_1}{\phi}  \quad \ddoc{a_2}{t_2}{\phi}
}$ 
\qquad \quad
$
\infer[\disc'_1 {[*]}]{
\Theory \cup \{(t_2:\neg \phi)_{r_1}, (t_1: \neg \phi)_{r_2}\}
}{
(t_1:\phi)_{r_1}  \quad (t_2: \phi)_{r_2}
}$\\

\\
$
\infer[\disc''_1 { [*]}]{
\Theory \cup \{(t_2:\neg \phi)_{r_1, r_i, \cdots, r_j}, (t_1: \neg \phi)_{r_2, r_m, \cdots, r_n}\}
}{
(t_1:\phi)_{r_1, r_i, \cdots, r_j}  \quad (t_2: \phi)_{r_2, r_m, \cdots, r_n}
}$\\

\\
$\infer[\disc_2{[\circ]} ]{
\Theory \setminus \{\ddoc{a_2}{t}{\neg \phi}\}
}{
a_2 \trust_p a_1 & \ddoc{a_1}{t}{\phi} & \ddoc{a_2}{t}{\neg \phi}
}$
\qquad \quad
$\infer[\disc'_2 { } ]{
\Theory \setminus \{(t:\neg \phi)_{r_2}\}
}{
r_2 \prec r_1 & (t:\phi)_{r_1} & (t: \neg \phi)_{r_2}
}$\\

\\
$
\infer[\disc''_2 {[\star]} ]{
\Theory \setminus \{(t:\neg \phi)_{r_2, r_m, \cdots, r_n}, \Delta\}
}{
r_2 \prec r_1 & (t:\phi)_{r_1, r_i, \cdots, r_j} & (t: \neg \phi)_{r_2, r_m, \cdots, r_n} & \Delta
}$\\

\\
$\infer[\crule_C  { [*]}]{
\bot
}{
\ddoc{a}{t_1}{\phi} & \ddoc{a}{t_2}{\phi}
}

\qquad
\infer[\crule'_C  { [*]}]{
\bot
}{
(t_1:\phi)_{r, \cdots, r_i} & (t_2: \phi)_{r, \cdots, r_j}
}$\\

\\
$\infer[\crule_T]{
\bot
}{
a \trust_p a
}
\qquad \qquad 
\infer[\crule'_T]{
\bot
}{
r \prec r
}
\qquad \qquad
\infer[\crule_P]{
\bot
}{
t:\phi & t:\neg \phi 
}$\\
{\scriptsize
where $[*] = [t_1 \neq t_2]$, 
$[\circ] = [\phi \in \Lit_S, \, \phi \textrm{ is } p \textrm{ or } \neg p]$, and
$[\star] = [\Delta = \bigcup_{(t':\psi)_\varrho \in \Theory\, \textit{s.t.}\, r_2 \in \varrho} (t':\psi)_\varrho]$
}
\end{tabular}
}
}
\end{center}
\end{table}

In this section, we introduce the \emph{rewriting system} of $\Hist$, which, as proved in Theorem~\ref{theosoundness}, is sound. Given pieces of evidence, the rewriting system yields a consistent set of pieces of evidence by translating pieces of evidence into interpretations and reasonings, and resolving their discordances.
In particular, the rewriting system uses the tableau rules in Table~\ref{table:rules} and applies them via the procedure in Algorithm~\ref{algo2:order}: given a $\Hist$-theory $\Theory$, which is a non-empty set of formulas, the rewriting system generates a new set of formulas $\widehat{\Theory}$ that replaces $\Theory$, where the single rewritings correspond to interpretations and reasonings of the theory.
More specifically, the rewriting system takes a $\Hist$-theory of the first level and rewrites 
it into a $\Hist$-theory of the second and third level, until all the formulas are interpreted, by adding formulas to the theory or eliminating formulas from the theory, with the use of \emph{insertion} or \emph{elimination rules}.

The rules in Table~\ref{table:rules} have as premises (above the line) a set of formulas and a $\Hist$-theory $\Theory$, although we don't show $\Theory$ for readability, and as conclusion (below the line) $\Theory \cup \{\phi\}$ or $\Theory \setminus \{\phi\}$, depending on whether the rule is an insertion or elimination rule that respectively
 inserts or eliminates $\phi$.
The \emph{insertion} rule introduce formulas for resolving temporal discordances 
and 
interpreting pieces of evidence.
The \emph{elimination} rules resolve discordances of event occurrences 
by deleting formulas based on the trust relations among agents and reasonings.
The \emph{closure rules} are part of the elimination rules, and
discover discordances in $\Theory$ that cannot be solved,
eliminate all the formulas of the set $\Theory$ and give as result the empty set $\bot$.

\subsection{Rewriting Rules}\label{rules}
\label{sec:rewriting:temporal}

We now explain the rules in more detail, starting from  the transformation rules that transform the formulas into the various layer formulas. 

Rule $\trule_1$ transforms a simple evidence into a temporal formula 
of the interpretation layer, whereas $\trule'_1$ transforms formulas of the reasoning layer into a temporal formula of the interpretation layer. 
\begin{example}
An application of $\trule_1$ in our use example is:
\begin{displaymath}\footnotesize
\infer[\trule_1]{
\Theory\cup \{t_2:\mathit{SpeedTr(23MB/s)}\}
}{
\ddoc{\FireEye}{t_2}{\mathit{SpeedTr(23MB/s)}}
}   \tag*{$\Box$}
\end{displaymath}
\end{example}

Rule $\trule_2$ transforms a derived evidence into an interpretation formula and, if possible, also introduces new pieces of evidence. 
Thus, given a derived evidence $\idoc{\ddoc{a}{t}{\phi}}{\ddoc{a_1}{t_1}{\phi_1} \mid \dots \mid \ddoc{a_n}{t_n}{\phi_n}}_r$, $\trule_2$
inserts the temporal formula
$t_1:\phi_1 \wedge \ldots \wedge t_n:\phi_n \rightarrow_r t:\phi$ 
and all $\ddoc{a_i}{t_i}{\phi_i}$ for $\phi_i \in \Lit_S$ and $i \in \{1,\ldots, n\}$. 
Note that rule $\trule_2$ inserts
 in the theory only the simple pieces of evidence that were part of the reasoning, and not the derived ones, as we
expect the pieces of evidence that are part of their reasonings to be part of the theory too.
\begin{example}
In our example, $\trule_2$ is applied to all derived pieces of evidence given by the sources. 
When it applies to $\CrowdStrike$'s first evidence, it transforms only the simple evidence about the spear phishing campaign, but not the successful phishing evidence as it is a derived one. The same occurs to $\TheForensicator$'s first evidence where just $\mathit{MetaC}$ is introduced:
\begin{displaymath}\footnotesize
\renewcommand{\arraystretch}{2.5}
\begin{array}{c}
\infer[\trule_2]{
\Theory \cup \{\ddoc{\CrowdStrike}{t_1}{\mathit{SPhish}}\}\cup \{t_1: \mathit{SPhish} \wedge t_1: \mathit{SucPhis} \rightarrow_{r_1} t_1: \mathit{Attack}\}
}{
\idoc{\ddoc{\CrowdStrike}{t_1}{\mathit{Attack}}}{\ddoc{\CrowdStrike}{t_1}{\mathit{SPhish}} \mid \ddoc{\CrowdStrike}{t_1}{\mathit{SucPhish}}}_{r_1}
} \\
\infer[\trule_2]{
\Theory \cup \{\ddoc{\TheForensicator}{t_2}{\mathit{MetaC}}\} \cup \{t_2: \mathit{MetaC} \wedge t_2: \mathit{PhysA} \rightarrow_{r_3} t_2: \mathit{Attack}\}
}{
\idoc{\ddoc{\TheForensicator}{t_2}{\mathit{Attack}}}{\ddoc{\TheForensicator}{t_2}{\mathit{MetaC}} \mid \ddoc{\TheForensicator}{t_2}{\mathit{PhysA}}}_{r_3}
} \\
\infer[\trule_2]{
\Theory \cup \{\ddoc{\TheForensicator}{t_2}{\neg \mathit{SpeedTr(23MB/s)}}\}\cup \{t_2: \neg \mathit{SpeedTr(23MB/s)} \rightarrow_{r_4} t_2: \mathit{PhysA}\}
}{
\idoc{\ddoc{\TheForensicator}{t_2}{\mathit{PhysA}}}{\ddoc{\TheForensicator}{t_2}{\neg \mathit{SpeedTr(23MB/s)}}}_{r_4}
}
\end{array}
\end{displaymath}
Applying $\trule_2$ to the second evidence of $\CrowdStrike$ yields
$\Theory \cup \{\ddoc{\CrowdStrike}{t_1}{\mathit{LinkCl}},\, \ddoc{\CrowdStrike}{t_1}{\mathit{FFill}},\, \ddoc{\CrowdStrike}{t_1}{\mathit{DStolen}}\}$
and $t_1: \mathit{LinkCl} \wedge t_1:\mathit{FFill} \wedge t_1:\mathit{DStolen}  \rightarrow_{r_2} t_1: \mathit{SucPhish}$.  \hfill $\Box$
\end{example}

Rules $\secrule$ and $(\rightarrow')$ transform the interpretation  formulas introduced by $\trule_2$ into reasoning  formulas (derived evidence of the two types). 
\begin{example}
Applying $(\rightarrow)$ to $\CrowdStrike$'s derived pieces of evidence yields:
\begin{displaymath}\scriptsize
\infer[(\rightarrow)]{
\Theory \cup \{(t_1: \mathit{SucPhish})_{r_2}\}
}{
 t_1: \mathit{LinkCl} \wedge t_1:\mathit{FFill} \wedge t_1:\mathit{DStolen}  \rightarrow_{r_2} t_1: \mathit{SucPhish} \!
 & \!
 t_1:\! \mathit{LinkCl} & \! t_1: \! \mathit{FFill} & \! t_1: \! \mathit{DStolen}
}
\end{displaymath}
Applying $(\rightarrow')$ to the second type of derived evidence for $\CrowdStrike$ yields
\begin{displaymath}\footnotesize
\infer[(\rightarrow')]{
\Theory \cup \{(t_1: Attack)_{r_1,r_2}\}
}{
 t_1: \mathit{SPhish} \wedge t_1: \mathit{SucPhish} \rightarrow_{r_1} t_1: \mathit{Attack}
 & t_1: \mathit{SPhish}
 & (t_1: \mathit{SucPhish})_{r_2}
 } \tag*{$\Box$}
\end{displaymath}
\end{example}

The $\trust$ and $\prec$ relations are transitive ones. $\trustTrans$ and $\RtrustTrans$ extend the trust relations between agents and reasonings, e.g., if $a_1$ is less trusted than $a_2$ with respect to $p$, and $a_2$ is less trusted than $a_3$ with respect to $p$, then $\trustTrans$ inserts into the theory the conclusions that $a_1$ is less trusted than $a_3$ with respect to $p$ (the same applies to $\prec$ with $\RtrustTrans$).

The discordance resolution rules resolve temporal and factual discordances, where
events are instantaneous and not recurring.
A \emph{temporal discordance} about an event occurs when two agents state that it occurred in two different instants of time, e.g., Alice states that $x$ occurred at $t_1$ and Bob states that it occurred at $t_2$.
A \emph{factual discordance} about an event occurs when there are inconsistent statements about the occurrence of an event at an instant of time, e.g., Alice states that at $t$ occurred $p$ and Bob states that at $t$ did not occur $p$.

Rules $\disc_1$, $\disc'_1$ and $\disc''_1$ transform temporal discordances into factual ones, where $\disc_1$ works with simple pieces of evidence, $\disc'_1$ with derived pieces of evidence of the first type, and $\disc''_1$ with derived pieces of evidence of the second type (note that $\disc'_1$ is a special case of $\disc''_1$). 
Thus, if the $\Hist$-theory $\Theory$ contains the evidence belonging to two different agents about the same event $p$, occurring at two different instants, then the evidence of the occurrence or not of $p$ with respect to both agents and both instants of time are inserted in the theory. 

Rule $\disc_2$, $\disc'_2$ and $\disc''_2$ solve the factual discordances based on the relations of trust, where $\disc_2$ eliminates from the theory the evidence of the less trusted agent, whereas $\disc'_2$ and $\disc''_2$ eliminate the evidence of the less trusted reasoning. $\disc''_2$ eliminates also every evidence that has inside its reasoning the removed evidence, as captured by the side condition where $\Delta$ is the set of all derived pieces of evidence that have $r_2$ in their reasonings:
$\Delta = \bigcup_{(t':\psi)_\varrho \in \Theory\, \textit{s.t.}\, r_2 \in \varrho} (t':\psi)_\varrho]$,
where $\varrho = \{r_k, \cdots, r_l\}$.
\begin{example}
$\disc_2$  solves the discordance of the speed transfer:
\begin{displaymath}\scriptsize
\infer[\disc_2 ]{
\Theory \setminus \{\ddoc{\TheForensicator}{t_2}{\neg \mathit{SpeedTr(23MB/s)}}\}
}{
\begin{array}{c}
\TheForensicator \trust_{\mathit{SpeedTr(23MB/s)}} \FireEye \ \ \ddoc{\FireEye}{t_2}{\mathit{SpeedTr(23MB/s)}} \  \ \ddoc{\TheForensicator}{t_2}{\neg \mathit{SpeedTr(23MB/s)}}
\end{array}
}   \tag*{$\Box$}
\end{displaymath}
\end{example}

The rewriting system has five closure rules that correspond to five discordances that cannot be solved resulting in the empty theory $\bot$.
$\crule_C$ applies when an agent contradicts herself, 
$\crule'_C$ when a reasoning contradicts itself. 
$\crule_T$ and $\crule'_T$ apply when an agent/reasoning is more trusted than herself/itself (we avoid these types of conflicts in the semantics thanks to $\mathit{COND_1}$ and $\mathit{COND_2}$, where $\trust$ and $\prec$ are irreflexive).
Finally, $\crule_P$ captures contradictions of the second layer,
where two temporal formulas state the occurrence and non occurrence of
an event at the same instant of time. 
This discordance occurs when we were not able to solve it using the trust relations.

\begin{theorem}
	\label{theosoundness}
The rewriting system of $\Hist$ is sound.
\end{theorem}

\noindent
The proof of the theorem is in the Appendix.

\subsection{Rewriting Procedure}
\label{sec:rewriting-procedure}

We give a procedure that uses the rewriting rules to prove the satisfiability of a given $\Hist$-theory. This procedure defines an order of application of the rules that rewrites the $\Hist$-theory as defined in Algorithm~\ref{algo2:order}. Theorem~\ref{theorem:correct} tells us that the procedure is correct (the theorem is proved in the Appendix).
\begin{algorithm}[t!]\small
\caption{Algorithm for the Rewriting Procedure}
\label{algo2:order}
\begin{algorithmic}[1]
\WHILE{We can apply $\trustTrans$, $\RtrustTrans$ rules}
	\STATE Apply $\trustTrans$ and $\RtrustTrans$ rules
\ENDWHILE
\STATE{Apply $\crule_T$ and $\crule'_T$} 
\IF{We have $\bot$}
	\STATE{We do not have a model. Exit!}
\ENDIF
\WHILE{We can apply $\trule_2$ rule}
	\STATE Apply $\trule_2$ rule
\ENDWHILE
\WHILE{We can apply $\disc_1$, $\disc_2$ rule}
	\STATE Apply $\disc_1$, $\disc_2$ rule
\ENDWHILE
\STATE{Apply $\crule_C$} 
\IF{We have $\bot$}
	\STATE{We do not have a model. Exit!}
\ENDIF
\COMMENT{Solved the discordancies of simple evidences.}
\WHILE{We can apply $\trule_1$ rule}
	\STATE Apply $\trule_1$ rule
\ENDWHILE
\WHILE{We can apply $(\rightarrow)$ rule}
	\STATE Apply $(\rightarrow)$ rule
\ENDWHILE
\WHILE{We can apply $\disc'_1$, $\disc'_2$ rule}
	\STATE Apply $\disc'_1$, $\disc'_2$ rule
\ENDWHILE
\WHILE{We can apply $(\rightarrow')$ rule}
	\STATE Apply $(\rightarrow')$ rule
\ENDWHILE
\WHILE{We can apply $\disc''_1$, $\disc''_2$ rule}
	\STATE Apply $\disc''_1$, $\disc''_2$ rule
\ENDWHILE
\STATE{Apply $\crule'_C$} 
\IF{We have $\bot$}
	\STATE{We do not have a model. Exit!}
\ENDIF
\WHILE{We can apply $\trule'_1$ rule}
	\STATE Apply $\trule'_1$ rule
\ENDWHILE
\STATE{Apply $\crule_P$} 
\IF{We have $\bot$}
	\STATE{We do not have a model. Exit!}
\ENDIF
\end{algorithmic}
\end{algorithm}

Given a $\Hist$-theory, the procedure starts by generating all the trust
relations, applying ($\trustTrans$) and ($\RtrustTrans$). 
Any contradiction that exists between trust relations is immediately captured by $\crule_T$ and $\crule'_T$. 
$\trule_2$ is applied to transform any derived evidence into its interpretations.
If needed, $\disc_1$ and $\disc_2$ are applied.
At this point all possible simple pieces of evidence are generated.
Any contradiction between first layer formulas is captured by $\crule_C$.
Afterwards, $\trule_1$ transforms any simple evidence into second layer formulas that are used by ($\rightarrow$) to obtain
reasoning layer formulas. 
$\disc'_1$ and $\disc'_2$ are applied to solve discordances between reasoning layer formulas based on the reasonings' trust relations. 
The result of the previous rules is used by ($\rightarrow'$) to generate reasoning layer formulas from derived pieces of evidence of the second type.
If any discordance arises, it is solved by $\disc''_1$ and $\disc''_2$,
where rule $\disc''_2$ not only takes out the not preferred evidence, but also any derived evidence that uses it as a precondition.
If no contradiction between reasoning rules is captured by $\crule'_C$, then $\trule'_1$ transforms all reasoning layer formulas into interpretation layer ones.
If $\crule_P$ applies, then there is a contradiction and we have $\bot$,
else no further transformation can be done, and the resulting set of formulas is the model of $\Hist$-theory.

\begin{example}
By applying the procedure we find that $(\rightarrow)$ can be applied only to $\CrowdStrike$'s pieces of evidence as the derived ones of $\TheForensicator$ are missing their premises, removed by $\disc_2$. 
Applying $\trule'_1$ yields $t_1:\mathit{Attack}$ and the analyst concludes that the attack occurred during March/April 2016. \hfill$\Box$
\end{example}

\begin{theorem}\label{theorem:correct}
The order of the rules in Algorithm~\ref{algo2:order} used by the rewriting procedure is correct.
\end{theorem}

\section{A Detailed Case Study: Attribution of a Cyber-Attack}
\label{sec:example}

The Evidence Logic $\Hist$ can be used in diverse application areas where there is a need to analyze and reason about conflicting data/knowledge. 
In this section, as a concrete proof of concept to show how to apply $\Hist$ during the investigations on a cyber-attack, we discuss a cyber-forensics case study in which the analyst needs to collect various pieces of evidence and analyze them to decide who performed the attack; this process
is called \emph{attribution} of the attack to a particular entity. 

As we remarked above, forensics investigations typically produce an enormous amount of evidence that need to be analyzed. The pieces of evidence are produced/collected by various sources, which can be humans (e.g., another analyst) or forensic tools such as intrusion detection system (IDS), traceback systems, malware analysis tools, and so on. 
The analyst trusts more some sources than others for particular pieces of evidence, e.g., source $S_1$ is more trusted than source $S_5$ for
attack similarity as tool $S_1$ specializes in malware analysis whereas tool $S_5$ specializes in deleted data. 
The collected evidence  
can be conflicting or bring  to conflicting results. The $\Hist$ Logic represents the evidence, together with its sources and relations of trust, and reasons about it, by eliminating the conflicting evidence and helping the analyst during the analysis process.

Suppose the analyst has collected (from analysts $A_1, \ldots, A_4$ and sources $S_1$, $S_2$, $S_3$) and is analyzing, using $\Hist$, the following pieces of evidence, representing events related to the attack that occurred (for the sake of space, we give a simplified but realistic version of the evidence that can be easily extended).

\begin{displaymath}
\footnotesize
\begin{array}{ll}
A_1: t:& \mathit{Culprit}(C, \mathit{Attack})[S_1: t: \mathit{sIP(Attack, IP)} \mid S_1: t: \mathit{Geoloc(IP, C)} \mid  S_2: t: \\
& \mathit{Cap(C, Attack)}]_{r_1}\\
A_2: t:& \mathit{Culprit(C, Attack)}[S_2: t: \mathit{Motive(C, Attack)} \mid S_2: t: \mathit{Cap(C, Attack)}]_{r_2}\\
A_3: t:&  \neg \mathit{Culprit(C, Attack)}[S_3: t: \neg \mathit{Cap(C, Attack)} \mid S_4: t: \neg \mathit{Fin(C, Attack)}]_{r_3}\\
A_4: t:& \neg \mathit{Culprit(C, Attack)}[S_1: t: \mathit{sIP(Attack, IP)} \mid S_1: t: \mathit{Geoloc(IP, C)} \mid S_7: t: \\
& \mathit{Spoofed(IP)}]_{r_4}\\
S_1: t:& \mathit{sIP(Attack, IP)}\\
S_1: t:& \mathit{Geoloc(IP, C)}\\
S_2: t:& \mathit{Cap(C, Attack)}[S_6: t_1: \mathit{Admit(C, Attack')} \mid S_1: t: \mathit{Sim(Attack, Attack')}]_{r_5}\\
S_3: t:& \neg \mathit{Cap(C, Attack)}[S_6: t_1: \mathit{Admit(C, Attack')} \mid S_5: t: \neg \mathit{Sim(Attack, Attack')}]_{r_6}\\
S_2: t:& \mathit{Motive(C, Attack)}[S_5: t: \mathit{EConf(C, Victim)}]_{r_7}\\
& 
S_5 \trust_{\mathit{Sim}} S_1\qquad r_1 \prec r_4 \qquad 
r_4 \prec r_2 \qquad 
r_2 \prec r_3 
\end{array}
\end{displaymath}
$\mathit{sIP(Attack, IP)}$ means that the $\mathit{Attack}$ came from $\mathit{IP}$; $\mathit{Geoloc(IP, C)}$ that $\mathit{IP}$ has country $C$ as geographical location; $\mathit{Cap(C, Attack)}$  that country $C$ has the capability of conducting the $\mathit{Attack}$. 
Analyst $A_1$ states that (based on reasoning $r_1$), given country $C$ 
is capable of performing the $\mathit{Attack}$ (stated by $S_2$) and it came from $\mathit{IP}$ located in $C$ (stated by $S_1$), then $C$ performed (is the culprit of) the attack, i.e., $\mathit{Culprit(C, Attack)}$. 
$A_2$ states that $C$ is the culprit (based on $r_2$),
as it has the capability of and the motive $\mathit{Motive(C, Attack)}$ for performing it (both stated by $S_2$). 
$A_3$ states that $C$ is not the culprit (based on $r_3$), as it is not capable of 
and (as stated by $S_4$) does not have the financial resources $\mathit{Fin(C, Attack)}$ for commissioning the attack.
$A_4$ states that $C$ is not the culprit (based on $r_4$), 
as the $\mathit{IP}$'s are $\mathit{Spoofed}$ (stated by $S_7$), so
their geolocation cannot be used.
Source $S_1$ states that the $\mathit{IP}$ from which the attack originated is located in $C$.
$S_2$ states that $C$ is capable (based on $r_5$), as $C$  admitted to be the culprit of a previous attack, i.e., $\mathit{Admit(C,Attack')}$, at $t_1$ (stated by $S_6$), and the latter is similar ($\mathit{Sim}$) to $\mathit{Attack}$ (stated by $S_1$). 
$S_3$ states that $C$ is not capable of performing  $\mathit{Attack}$
(based on $r_6$), as $\mathit{Attack'}$ that $C$ admitted to have performed, is not similar to $\mathit{Attack}$ (stated by $S_5$). 
$S_2$ states that $C$ has motive for the attack (based on $r_7$), as
$C$ has an economical conflict $\mathit{EConf}$ with the attack $\mathit{Victim}$ (stated by $S_5$).
Our analyst trusts more source $S_1$ than $S_5$ for the similarity between attacks, and reasoning $r_3$ more than $r_2$, $r_2$ more than $r_4$ and $r_4$ more than $r_1$.

The simple pieces of evidence of this use case are:
\begin{displaymath}
\footnotesize
\begin{array}{ll}
\Vars_S = & \{\mathit{sIP(Attack, IP)}, \mathit{Geoloc(IP, C)}, \mathit{Fin(C, Attack)}, \mathit{Admit(C, Attack')},\\ &
 \ \mathit{Sim(Attack, Attack')}, \mathit{Spoofed(IP)}, \mathit{EConf(C, Victim)}\}.
\end{array}
\end{displaymath}

\begin{figure}[t!]
\centering
\tiny
\begin{equation}\label{eq1}
\infer[\trule_2]{
\begin{array}{c}
\Theory \cup\{S_1:t: sIP(Attack, IP), S_1:t: Geoloc(IP, C)\} \cup \\ \{\ t:sIP(Attack, IP) \wedge t:Geoloc(IP, C) \wedge t:Cap(C, Attack)\rightarrow_{r_1} t:Culprit(C, Attack) \}
\end{array}
}{
A_1:t:Culprit(C, Attack)[S_1:t:sIP(Attack, IP) \mid S_1:t: Geoloc(IP, C) \mid S_2:t: Cap(C, Attack)]_{r_1}}
\end{equation}
\begin{equation}
\infer[\trule_2]{
\Theory \cup \{\ 
t: Motive(C, Attack) \wedge t: Cap(C, Attack)\rightarrow_{r_2} t: Culprit(C, Attack)\}
}{
A_2:t: Culprit(C, Attack)[S_2:t: Motive(C, Attack) \mid S_2:t: Cap(C, Attack)]_{r_2}}
\end{equation}
\begin{equation}
\infer[\trule_2]{
\begin{array}{c}
\Theory \cup\{S_4:t: \neg Fin(C, Attack)\} \cup \\
\{t: \neg Cap(C, Attack) \wedge t: \neg Fin(C, Attack)\rightarrow_{r_3} t: \neg Culprit(C, Attack)\}
\end{array}
}{
A_3:t: \neg Culprit(C, Attack)[S_3:t: \neg Cap(C, Attack) \mid S_4:t: \neg Fin(C, Attack)]_{r_3}}
\end{equation}
\begin{equation}
\infer[\trule_2]{
\begin{array}{c}
\Theory \cup\{S_6: t_1: \mathit{Admit}(C, Attack'), S_1: t: Sim(Attack, Attack')\} \cup \\ \{\ t_1: Admit(C, Attack') \wedge t: Sim(Attack, Attack') \rightarrow_{r_5} t: Cap(C, Attack)\}
\end{array}
}{
S_2: t: Cap(C, Attack)[S_6: t_1: \mathit{Admit}(C, Attack') \mid S_1: t: Sim(Attack, Attack')]_{r_5}}
\end{equation}
\begin{equation}
\infer[\trule_2]{
\begin{array}{c}
\Theory \cup\{S_6: t_1: Admit(C, Attack'), S_5: t: \neg Sim(Attack, Attack')\} \cup \\\{\ t_1: \mathit{Admit}(C, Attack') \wedge t: \neg Sim(Attack, Attack') \rightarrow_{r_6} t: \neg Cap(C, Attack)\}
\end{array}
}{
S_3: t: \neg Cap(C, Attack)[S_6: t_1: Admit(C, Attack') \mid S_5: t: \neg Sim(Attack, Attack')]_{r_6}}
\end{equation}
\begin{equation}
\infer[\trule_2]{
\begin{array}{c}
\Theory \cup\{S_1:t: sIP(Attack, IP), S_1:t: Geoloc(IP, C), S_7:t: Spoofed(IP)\} \cup \\ \{\ t:sIP(Attack, IP) \wedge t:Geoloc(IP, C) \wedge t: Spoofed(IP)\rightarrow_{r_4} t:\neg Culprit(C, Attack) \}
\end{array}
}{
A_4:t:\neg Culprit(C, Attack)[S_1:t:sIP(Attack, IP) \mid S_1:t: Geoloc(IP, C) \mid S_7:t: Spoofed(IP)]_{r_4}}
\end{equation}
\begin{equation}\label{eq7}
\infer[\trule_2]{
\Theory \cup\{S_5: t:EConf(C, Victim)\} \cup \{\ t: EConf(C, Victim) \rightarrow_{r_7} t: Motive(C, Attack)\}
}{
S_2: t: Motive(C, Attack)[S_5:t:EConf(C, Victim)]_{r_7}}
\end{equation}
\begin{equation}\label{eqd}
\infer[\disc_2]{
\Theory \setminus \{S_5: t: \neg Sim(Attack, Attack')\}
}{
S_5 \trust_{Sim} S_1 & S_1: t: Sim(Attack, Attack') & S_5: t: \neg Sim(Attack, Attack')
}
\end{equation}
\begin{equation}\label{eq8}
\infer[(\rightarrow)]{
\Theory \cup \{(t: Cap(C, Attack))_{r_5}\}
}{\begin{array}{c}
t_1:\mathit{Admit}(C, Attack')  \qquad  t: Sim(Attack, Attack') \\
t_1:\mathit{Admit}(C, Attack') \wedge t: Sim(Attack, Attack') \rightarrow_{r_5} Cap(C, Attack)
\end{array}
}
\end{equation}
\begin{equation}\label{eq10}
\infer[(\rightarrow)]{
\Theory \cup\{(t:\neg Culprit(C, Attack))_{r_4}\}
}{
\begin{array}{c}
t: sIP(Attack, IP) \qquad t: Geoloc(IP, C) \qquad t: Spoofed(IP)\\
t:sIP(Attack, IP) \wedge t:Geoloc(IP, C) \wedge t: Spoofed(IP)\rightarrow_{r_4} t:\neg Culprit(C, Attack)  \end{array}
}
\end{equation}
\begin{equation}\label{eq11}
\infer[(\rightarrow)]{
\Theory \cup\{ (t: Motive(C, Attack))_{r_7}\}
}{
t: EConf(C, Victim) \rightarrow_{r_7} t: Motive(C, Attack) &
 t:EConf(C, Victim)}
\end{equation}
\begin{equation}\label{eq12}
\infer[(\rightarrow')]{
\Theory \cup \{(t: Culprit(C, Attack))_{r_1,r_5}\}
}{
\begin{array}{c}
t:sIP(Attack, IP) \qquad t: Geoloc(IP, C) \qquad (t:Cap(C, Attack))_{r_5}\\
t:sIP(Attack, IP) \wedge t: Geoloc(IP, C) \wedge t:Cap(C, Attack) \rightarrow_{r_1} t:Culprit(C, Attack) \end{array}
}
\end{equation}
\begin{equation}\label{eq13}
\infer[(\rightarrow')]{
\Theory \cup \{(t: Culprit(C, Attack))_{r_2,r_7,r_5}\}
}{
\begin{array}{c}
(t:Motive(C,Attack))_{r_7} \qquad (t:Cap(C, Attack))_{r_5}\\
t:Motive(C, Attack) \wedge t:Cap(C, Attack) \rightarrow_{r_2} t:Culprit(C, Attack)
\end{array}
}
\end{equation}
\begin{equation}\label{eq14}
\infer[\disc''_2 { } ]{
\Theory \setminus \{(t:Culprit(C, Attack))_{r_1, r_5}\}
}{
r_1 \prec r_4 & (t: Culprit(C, Attack))_{r_1,r_5} & (t:\neg Culprit(C, Attack))_{r_4}
}
\end{equation}
\begin{equation}\label{eq15}
\infer[\disc''_2 { } ]{
\Theory \setminus \{(t:\neg Culprit(C, Attack))_{r_4}\}
}{
r_4 \prec r_2 & (t: Culprit(C, Attack))_{r_2,r_7,r_5} & (t:\neg Culprit(C, Attack))_{r_4}
}
\end{equation}
\caption{Application of the Rewriting Procedure}\label{fig:L2rule}
\end{figure}

Let us now apply $\Hist$'s rewriting procedure. 
We start with rules $\trustTrans$ and $\RtrustTrans$: the first cannot be applied, the second yields $r_1 \prec r_2$, $r_1 \prec r_3$ and $r_4 \prec r_3$.
Neither $\crule_T$ nor $\crule'_T$ can be applied. 
We show the application of $\trule_2$ to the pieces of evidence 
 in (\ref{eq1})--(\ref{eq7}) in Fig.~\ref{fig:L2rule}.
In~(\ref{eqd}) rule $\disc_2$ eliminates $S_5: t: \neg \mathit{Sim(Attack, Attack')}$.
No contradiction is captured by $\crule_C$, and
$\trule_1$ transforms all first layer formulas into second layer ones:
\begin{displaymath} 
\footnotesize
\begin{array}{ll}
\Theory \cup & \{t: \mathit{sIP(Attack, IP)}, t: \mathit{Geoloc(IP, C)}, t: \neg \mathit{Fin(C, Attack)}, t: \mathit{Spoofed(IP)}\\
& \  
t_1: \mathit{Admit(C, Attack')}, t: \mathit{Sim(Attack, Attack')}, t: \mathit{EConf(C, Victim)}\}.
\end{array}
\end{displaymath}
(\ref{eq8})--(\ref{eq11}) show applications of $(\rightarrow)$ to any evidence that has its premises in the theory. 
$\disc'_1$ and $\disc'_2$ cannot be applied as there is no temporal/factual discordance between derived pieces of evidence of the first type. 
Applying $(\rightarrow')$ produces derived pieces of evidence of the second type for $A_1$ and $A_2$ as shown in (\ref{eq12})--(\ref{eq13}).
$A_3$'s evidence is not derived as $C$ is capable to perform the attack.
Rule $\disc''_1$ cannot be applied. 
Rule $\disc''_2$ is applied, as shown in (\ref{eq14})-(\ref{eq15}), to
the conflicting pieces of evidence where the reasonings' trust relations apply.
Finally, $\trule'_1$ transforms all third layer formulas into second layer ones:
\begin{displaymath} 
\footnotesize
\begin{array}{ll}
\Theory \cup & \{t:sIP(Attack, IP), t:Geoloc(IP, C), t: \neg Fin(C, Attack),t: Spoofed(IP)\\
& \  
t_1: \mathit{Admit(C, Attack')},
  t: \mathit{Sim(Attack, Attack')},  t: \mathit{Cap(C, Attack)},\\
  & \ t: \mathit{EconfConflict(C, Victim)}, t: \mathit{Motive(C, Attack)}, t: \mathit{Culprit(C,Attack)}\}.
\end{array}
\end{displaymath}
The analyst, given the result of the procedure, concludes that the culprit of the $\mathit{Attack}$ is $C$.

The question of ``who performed the attack'' is, in general, not an easy one to answer, but we believe that $\Hist$ can be successfully used to analyze and filter the large amount of cyber-forensics evidence that an analyst needs to deal with. At the very least, $\Hist$ allows an analyst to perform a first, formal filtering of the evidence and obtain different plausible conclusions, which the analyst can then further investigate.

\section{Related Work and Concluding Remarks}
\label{sec:related}

When we introduced $\Hist$ and discussed how it allows analysts to reason about simple and derived evidence given by different sources, we deliberately did not use the notion of ``belief''.
We chose to do so as the main scope of our work is not to consider modalities of knowledge or belief, but to introduce a procedure that analyzes and filters the potentially enormous amount of forensics evidence, eliminate discordances and reach conclusions. 
The notion of evidence (both simple and derived evidence) can be represented quite naturally as agents' beliefs and, in fact, the reasoning process in $\Hist$ could be considered a belief revision process. However, our procedure, differently from the belief revision process, uses a monotonic reasoning, does not distinguish between beliefs and knowledge, is based on the notion of trust,
and does not apply the principle of minimal change.

\emph{Belief revision} is the process of integrating 
new information with existing beliefs or knowledge
\cite{Benthem07,Van07,BaltagS06a,AugustoS01,Dix2016}.
It is performed based on the knowledge and beliefs
of the user and the beliefs of other agents announced, privately~\cite{AgotnesBDS10,BalbianiGHL11} or 
publicly~\cite{Plaza07,Balbiani07a}, and it uses non-monotonic reasoning.
In our approach, we use \emph{monotonic reasoning} as we expect 
only the final set, that represents our theory, to be consistent.
Our procedure deals with conflicting pieces of evidence,
which are analyzed by expanding or contracting 
the evidence set. 
In case of unsolved inconsistencies,
our theory is empty. 
The procedure does not
incorporate every incoming information 
in the evidence set, but rather the new evidence is included or
not depending on the trust relations. 
This is different from the classical AGM belief revision~\cite{AGM}, where 
the \emph{principle of minimal belief change} applies. 

Our analysis can be seen as a revision procedure, where we do not distinguish between
beliefs and knowledge. 
 Thus, all the pieces of evidence can be treated as beliefs, and
there is no space for personal or common belief/knowledge. 
Some works have considered belief revision that uses relation of trust between agents 
\cite{LoriniJP14,HunterB15,BarberK00,Alechina08}. 
However, not much effort has been devoted to working with a relation of trust relative to the reasoning used 
to arrive to certain conclusions.  
Our trust relations do not have a grading system, like the one in~\cite{LoriniJP14},
which is difficult to define 
for cyber-forensics data, but use comparable trust between the sources based on the evidence, similar to~\cite{HunterB15}, where a notion of trust restricted to a domain of expertise is used.
As future work, we plan to use Bayesian belief networks~\cite{BarberK00},
and the Dempster-Shafer theory to
quantify the level of trust for the evidence, and
to enrich our framework with trust reinforcement mechanisms.

To the best of our knowledge, the only attempt at using belief revision
during  
cyber-attacks' investigations is~\cite{ShakarianSMP15,ShakarianSMPPFA16}, where a \emph{probabilistic structure argumentation} framework is used
to analyze contradictory and uncertain data.
Our procedure does not deal with probabilities, but with preferences between sources and reasoning rules. We believe this to be a more accommodating approach, especially for the main use case, investigations of cyber-attacks, where calculating and revising 
 probabilities is resource consuming. 
  The framework of~\cite{ShakarianSMP15,ShakarianSMPPFA16} allows attackers to use, during the deceptive attempts,
the well-known \emph{specificity criteria}, 
 i.e., the rule that uses more pieces of evidence 
 is preferable.  
We avoid this type of deceptive attempts as the trust relations are given by the analyst.

$\Hist$ is based on LTL. 
Another approach is to use 
\emph{Temporal Defeasible Logic}~\cite{AugustoS01}, where knowledge is represented as norms with temporal scope~\cite{GovernatoriT07}.
For the sake of simplicity, our stream of time is discrete and provided initially. 
As future work, we plan to consider the flow of time as not provided and as non-discrete in order to
have temporal relations between labels that represent the instants of time.

Another distinctive feature of our approach with respect to the rest of the literature that focuses on agents' trust relations and their reputation systems is the fact that we engage not only with the trust between agents,  but also with the reasoning behind the evidence.
Hence, even when a particular agent is not trusted, 
if the reasoning behind the evidence is sound, we might take it into account. The notion of trust, also seen as preference, is subjective to the analyst, and we assume that agents are sincere, and thus 
share all their information. 
As future work, we plan to incorporate both a reputation revision process, where the trustworthiness and reliability of the sources is analyzed and revised based on past experience, and private/public announcements. Finally, on the theoretical side, we plan to investigate the completeness of the rewriting system and algorithm, whereas on the practical side, we plan to fully automate our analysis process and to perform an evaluation analysis on real evidence of cyber-attacks.

\section*{Acknowledgments}
Erisa Karafili was supported by the European Union's H2020 research and innovation programme under the Marie Sk\l{}odowska-Curie grant agreement No.~746667.

\bibliographystyle{splncs04}
\bibliography{hist}

\appendix
\section{Appendix: Soundness of the Rewriting System and Correctness of the Algorithm}
\label{sec:sound}

In this appendix, we prove the soundness of the rewriting system of $\Hist$ and the correctness of Algorithm~\ref{algo2:order}.
Given a theory $\Theory$ and $\Hist$'s  rewriting
system, the application of at least one of its closure rules 
generates an empty set. In fact, every theory that contains $\bot$ is 
equivalent to the empty theory. 
When the input theory is not empty and has \emph{no contradiction}, 
then the theory rewritten by $\Hist$ should give as result 
a non empty theory.

As usual in \emph{tableau rewriting systems}, we define three fundamental notions: open, closed, and exhausted theories. A theory is \emph{closed} when it contains a contradiction and it is \emph{open} when it does not. A theory is \emph{exhausted} when it is a fixpoint with respect to the rewriting process, i.e., by applying the rewriting system to an exhausted theory $\Theory$, we always obtain $\Theory$. Under the grounded semantics introduced in Section~\ref{sec:semantics}, 
we prove the soundness of the rewriting system by showing that 
open theories have models under the semantics, and closed ones have not. Thus, when we find an open and exhausted theory, we can prove the existence of a model.

We show now that the rules that rewrite a theory $\Theory$ into $\widehat{\Theory}$ without introducing $\bot$ constitute by themselves a sound system. The proofs of Lemma~\ref{verybasiclemma} and Lemma~\ref{basiclemma} are straightforward and are omitted for the sake of space.

\begin{lemma}\label{verybasiclemma}
	If a satisfiable $\Hist$-theory $\Theory$ is rewritten into an exhausted theory $\widehat{\Theory}$,
	without using the closure rules, 
	then $\widehat{\Theory}$ entails consequence $C$ only when $C$ is a
	consequence of $\Theory$.
	\end{lemma}

\begin{lemma}\label{basiclemma}
	If an unsatisfiable $\Hist$-theory $\Theory$ is rewritten into an exhausted theory $\widehat{\Theory}$,
	then $\widehat{\Theory}$ is empty. 
\end{lemma}

\begin{lemma}\label{lemma:open}
Given a satisfiable theory $\Theory$,
the rewriting system $\Hist$ rewrites the theory in an open and exhausted one.
\end{lemma}
\begin{proof}
This lemma is proved by contradiction. Assume that $\Theory$ is non empty and satisfiable, and is rewritten by $\Hist$ into an exhausted closed theory $\widehat{\Theory}$.
Starting from a satisfiable theory $\Theory$ there are five cases of
rewriting it in a contradictory theory that gives as result $\bot$. 
In the definition of model in Section~\ref{sec:semantics}, we introduced four conditions that constrain the behavior of interpretations. Below we provide the complete analysis only for the first case (that is provided as a consequence of $\mathit{COND}_1$). The other cases are a natural extension of this one and are omitted for the sake of space.
The first case occurs when applying the $\crule_C$ rule.
We have that:
$a^\Interp (t_1, p) = \mathit{True}$, and 
$a^\Interp (t_2, p) = \mathit{True}$.
$\mathit{COND}_1$ implies that a propositional variable referred to an agent $a$ can be true in only one instant of time, thus, $\Theory$ is not satisfiable.
 \hfill $\blacksquare$
\end{proof}
 
Lemma~\ref{lemma:open} introduces a result for $\Hist$-soundness as it guarantees that irregardless of the order in which we apply the rules, we catch a contradiction at a given point. Thus, as a direct consequence of the grounded semantics, of Lemma~\ref{basiclemma} and of
Lemma~\ref{lemma:open} we obtain Theorem~\ref{theosoundness}.

When a sound rewriting system exists in a logical reasoning system
we always have a method to deliver satisfiability check of a theory. In this case, we apply the rules to a theory until we reach a fixpoint. 
If we aim at developing an effective method, however, we need to provide a proof of termination for such a method.
For the rewriting system of $\Hist$ we can prove 
that 
a simple approach, based on the execution of the rules in a given order, is sufficient to provide an effective method for satisfiability checking. This is the result of correctness of Algorithm~\ref{algo2:order}. Firstly, in Lemma~\ref{lemma:contradiction}, we prove that the existence of $\bot$ in a theory, if not introduced by default, is the consequence of the application of the rules in a specific order.

\begin{lemma}\label{lemma:contradiction}
If a satisfiable $\Hist$-theory $\Theory$ 
is rewritten into a contradictory $\widehat{\Theory}$, then we have:
\vspace{-5pt}
\begin{description}
\item[\ \textit{1.}] rewritten the theory by using $\trule_1$ before $\disc_1$ and $\disc_2$, or
\item[\ \textit{2.}] rewritten the theory by using $\trule'_1$ before $\disc'_1$, $\disc''_1$, $\disc'_2$ and $\disc''_2$, or
\item[\ \textit{3.}] rewritten the theory by using $\secrule$ before $\disc_1$ and $\disc_2$, or
\item[\ \textit{4.}] rewritten the theory by using $(\rightarrow')$ before $\disc'_1$ and $\disc'_2$, or
\item[\ \textit{5.}] applied ($\trustTrans$) after $\disc_1$ and $\disc_2$, and
\item[\ \textit{6.}] applied $(\RtrustTrans)$ after $\disc'_1$, $\disc''_1$, $\disc'_2$ and $\disc''_2$.
\end{description}
\end{lemma}
\begin{proof}
There are two cases when $\widehat{\Theory}$ is empty: either
(1) the theory $\Theory$ is empty or (2) a closure rule was used.
The first case is not possible by definition of the theory, as we assume that $\Theory$ is not empty. The second case occurs if at least one of the five closure rules applied. 
Suppose by contradiction that rule $\crule_T$ or $\crule'_T$ is used to compute the contradictory $\widehat{\Theory}$. The application of this rule leads to a contradiction as $\Theory$ is satisfiable, 
whilst $\crule_T$ or $\crule'_T$ are applied when there is a contradiction in the theory. The same applies for rules $\crule_C$ and $\crule'_C$.

 Suppose, ad absurdum, that $\crule_P$ leads to a contradictory $\widehat{\Theory}$. The premises of $\crule_P$ are obtained using rules $\trule_1$, $\trule'_1$,
 $\trule_2$, $\secrule$ and $(\rightarrow')$.
 The first case for having a contradiction captured by $\crule_P$
 is when $\trule_1$ is applied before  $\disc_1$ and $\disc_2$.
 This happens because a contradiction is found that in fact
 was solved by  $\disc_1$ and $\disc_2$, as $\Theory$ is satisfiable.
 The second case for having a contradiction captured by $\crule_P$
 is when $\trule'_1$ is applied before  $\disc'_1$, $\disc''_1$, $\disc'_2$ and $\disc''_2$.
 This happens because a contradiction is found that in fact
 was solved by  $\disc'_1$, $\disc''_1$, $\disc'_2$ and $\disc''_2$, as $\Theory$ is satisfiable.
 The third case for having a contradiction captured by $\crule_P$
 is when $\secrule$ is applied before  $\disc_1$ and $\disc_2$.
 This happens because the formulas that were introduced produce the contradictions that in fact were solved by  $\disc_1$ and $\disc_2$, as $\Theory$ is satisfiable.
 The fourth case is similar to the third and occurs when $(\rightarrow')$ is applied before $\disc'_1$ and $\disc'_2$.
 The fifth case for having a contradiction captured by $\crule_P$
 is when $\trustTrans$ is applied after $\disc_1$ and $\disc_2$.
 This happens because the contradictions found could be solved by $\disc_1$ and $\disc_2$ if the $\trustTrans$ rule was applied before,
 as $\Theory$ is satisfiable.
 The sixth case for having a contradiction captured by $\crule_P$
 is when $\RtrustTrans$ is applied after $\disc'_1$, $\disc''_1$, $\disc'_2$ and $\disc''_2$.
 This happens because the contradictions found could be solved by $\disc'_1$, $\disc''_1$, $\disc'_2$ and $\disc''_2$ if the $\RtrustTrans$ rule was applied before, as $\Theory$ is satisfiable.
 \hfill $\blacksquare$
\end{proof} 

We are now able to prove that the rewriting procedure
introduced in Section~\ref{sec:rewriting-procedure} establishes 
satisfiability as defined
in Definition~\ref{semantic}.
We prove that the provided \emph{specific} order of application of the rewriting rules determines the existence of a model.
Given Theorem~\ref{theosoundness}, we prove that the rewriting given by Algorithm~\ref{algo2:order} is exhausted.
Theorem~\ref{theorem:correct} follows by applying Lemma~\ref{lemma:open} and Lemma~\ref{lemma:contradiction}.

\begin{proof}[Theorem~\ref{theorem:correct}]
Based on the semantics introduced in Section~\ref{sec:semantics} and given that 
Algorithm~\ref{algo2:order}
 applies the rules in 
the order specified in Lemma~\ref{lemma:contradiction}, we show that every theory that is unsatisfiable is rewritten by 
Algorithm~\ref{algo2:order}
 in a closed one, and consequently, every open theory resulting by the rewriting procedure, is also exhausted. We prove this by induction on the theory construction. 

The base cases occur for a relational formula, a simple evidence,
or a derived one. For lack of space, we omit the proofs as they follow quite straightforwardly by the definitions of relational formula, simple evidence, derived evidence, and the $\trust$ and $\prec$ relations. 

For the inductive step, we assume that $\Theory$ is formed by either $n$ relational formulas, or a blend of $n$ formulas,
and that we know that the claim
is true for $n-1$ formulas, and we show that the claim then holds also for $n$ formulas.  

Assume that $\Theory$ is formed by $n$ different relational formulas. 
The only rules that can be applied are $\trustTrans$ and $\RtrustTrans$, and the algorithm applies them. 
If $\Theory$ is unsatisfiable, then $\widehat{\Theory}$ is empty as rule $\crule_T$ or rule $\crule'_T$ capture any existing contradictions between
relational formulas. If $\Theory$ is satisfiable, then
$\widehat{\Theory}$ is open and exhausted as the algorithm has applied all possible rules. 

Assume that $\Theory$ is formed by $n$ different simple pieces of evidence. 
The algorithm first tries to apply $\disc_1$. 
If $\Theory$ is unsatisfiable and there are discordances, then 
$\widehat{\Theory}$ is empty, because the algorithm applies $\crule_C$.
If there are no discordances in $\Theory$, then the algorithm translates all the rules into second layer formulas,
by applying rule $\trule_1$. Since $\Theory$ is unsatisfiable, the algorithm applies the closing rule $\crule_P$
to capture the discordances between second layer formulas, and $\widehat{\Theory}$ is empty. 
If $\Theory$ is satisfiable, then the algorithm applies $\trule_1$, and $\widehat{\Theory}$
is open and exhausted as the algorithm has applied all possible rules. 

Assume that $\Theory$ is formed of $n$ different derived pieces of evidence. 
The algorithm tries to apply the rules in the following order:
$\trule_2$, $\disc_1$, $\trule_1$, $(\rightarrow)$, $\disc'_1$, $(\rightarrow')$, $\disc''_1$ and $\trule'_1$.
If $\Theory$ is unsatisfiable, then the algorithm applies one of the closing rules $\crule_C$, $\crule'_C$ and $\crule_P$
to capture the discordances between formulas of the different layers, and $\widehat{\Theory}$ is empty.
If $\Theory$ is satisfiable, then the algorithm yields a $\widehat{\Theory}$
that is open and exhausted as it has applied all possible rules. 

Assume that $\Theory$ is formed of $n$ different formulas (pieces of evidence and relational formulas).
The algorithm tries to apply all of its rules. 
If $\Theory$ is unsatisfiable, then the algorithm applies one of the closing rules
to capture the discordances between formulas of the different layers, and $\widehat{\Theory}$ is empty.
We know that our algorithm is able to capture all the contradictions, because 
the algorithm first applies all the rules that can surface all the possible contradictions and 
then it applies the appropriate closing rule.
If $\Theory$ is satisfiable, then the algorithm yields a $\widehat{\Theory}$
that is open and exhausted as it has applied all possible rules. 
\hfill $\blacksquare$
\end{proof}

\end{document}